\documentclass[journal]{IEEEtranTIE}
\usepackage{graphicx}
\pdfoutput=1
\usepackage{cite}
\usepackage{picinpar}
\usepackage{amsmath}
\usepackage{url}
\usepackage{flushend}
\usepackage[latin1]{inputenc}
\usepackage{colortbl}
\usepackage{soul}
\usepackage{multirow}
\usepackage{pifont}
\usepackage{color}
\usepackage{alltt}
\usepackage[hidelinks]{hyperref}
\usepackage{enumerate}
\usepackage{siunitx}
\usepackage{breakurl}
\usepackage{epstopdf}
\usepackage{pbox}
\usepackage{subfigure}

\newtheorem{assumption}{Assumption}
\newtheorem{lemma}{Lemma}

\newtheorem{theorem}{Theorem}

\newenvironment{proof}{{\indent \indent \it Proof:}}

\begin{document}
\title{	Finite-Time Gradient Descent-Based Adaptive Neural Network Finite-Time Control Design for Attitude Tracking of a 3-DOF Helicopter}

\author{
	\vskip 1em
	
	Xidong Wang

	\thanks{
					
		Xidong Wang is with the Research Institute of Intelligent Control and Systems, School of Astronautics, Harbin Institute of Technology, Harbin 150001, China (e-mail: 17b904039@stu.hit.edu.cn). 
	}
}

\maketitle
	
\begin{abstract}
This paper investigates a novel finite-time gradient descent-based adaptive neural network finite-time control strategy for the attitude tracking of a 3-DOF lab helicopter platform subject to composite disturbances. First, the radial basis function neural network (RBFNN) is applied to estimate lumped disturbances, where the weights, centers and widths of the RBFNN are trained online via finite-time gradient descent algorithm. Then, a finite-time backstepping control scheme is constructed to fulfill the tracking control of the elevation and pitch angles. In addition, a hybrid finite-time differentiator (HFTD) is introduced for approximating the intermediate control signal and its derivative to avoid the problem of "explosion of complexity" in the traditional backstepping design protocol. Moreover, the errors caused by the HFTD can be attenuated by the combination of compensation signals. With the aid of the stability theorem, it is proved that the closed-loop system is semi-globally uniformly ultimately boundedness in finite time. Finally, a comparison result is provided to illustrate the effectiveness and advantages of the designed control strategy.
\end{abstract}


{}

\definecolor{limegreen}{rgb}{0.2, 0.8, 0.2}
\definecolor{forestgreen}{rgb}{0.13, 0.55, 0.13}
\definecolor{greenhtml}{rgb}{0.0, 0.5, 0.0}

\section{Introduction}

\IEEEPARstart{I}{n} the past few years, the laboratory 3-DOF helicopter (Fig. 1) has attracted considerable research interests owing to the similar dynamics with the real one and facilitating to implement various control approaches \cite{1.Li2015F}. In recent years, numerous nonlinear and intelligent control schemes have been established to work out the attitude tracking problem of the helicopter platform \cite{1.Li2015F,2.chen2016N,3.zeghlache2017,4.Chen2018N,5.liu2019attitude,6.Yang2020F,7.li2021}. 

\begin{figure}[!t]\centering
	\includegraphics[width=8.5cm]{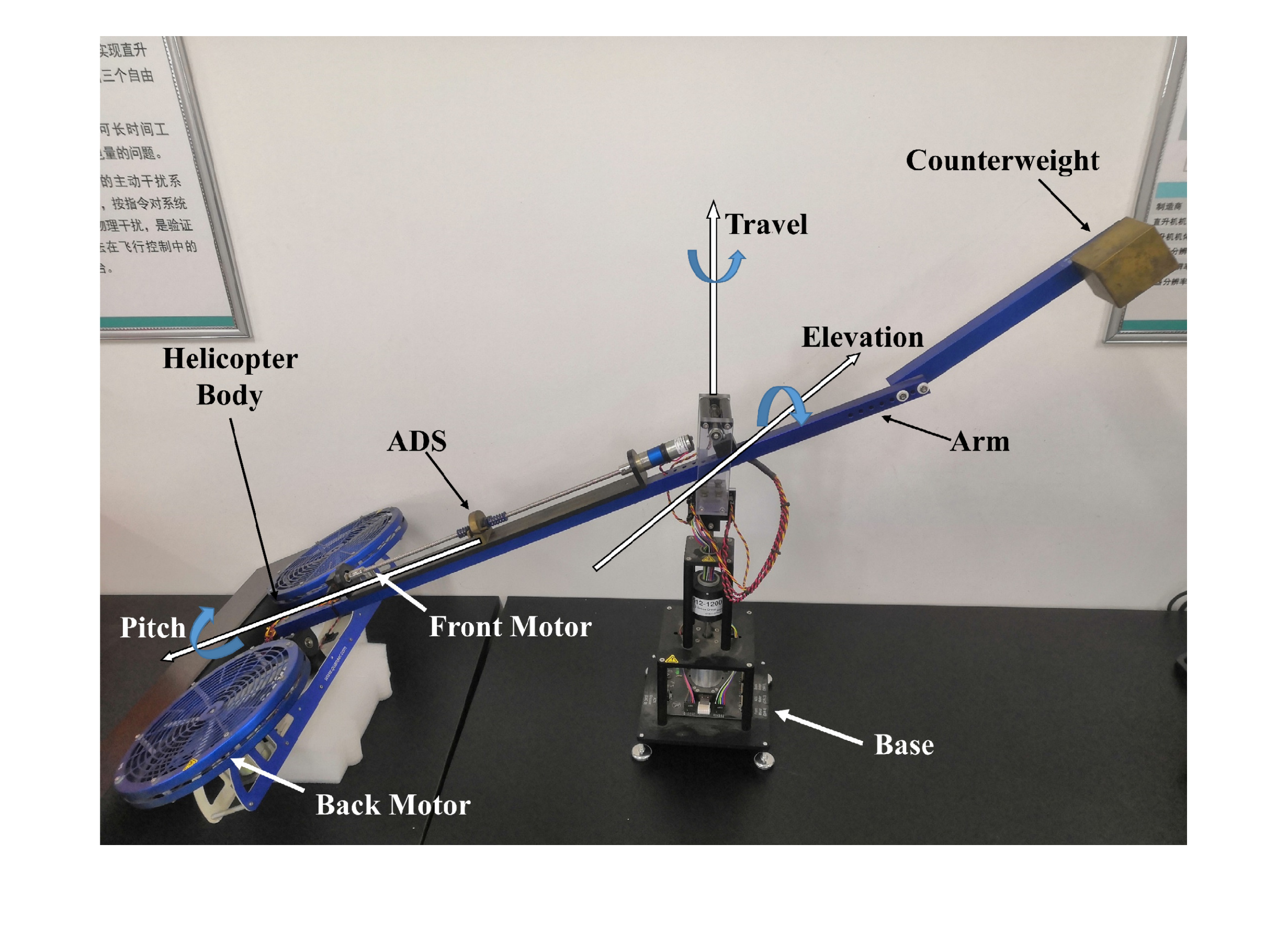}
	\caption{The laboratory 3-DOF helicopter platform}
\end{figure}

In \cite{1.Li2015F}, a nonlinear robust control method, augmented with UDE terms to approximate compound disturbances, was presented to achieve the semi-global asymptotic stability of a table-mount helicopter. The authors in \cite{2.chen2016N} developed a NN-based fault-tolerant control strategy for the helicopter platform under actuator faults, and the NN was applied to estimate the composite disturbance. By integrating the interval type-2 fuzzy logic into SMC, the developed control protocol in \cite{3.zeghlache2017} fulfilled the exponential convergence of helicopter platform tracking error. The authors in \cite{4.Chen2018N} addressed the attitude tracking problem of an experimental helicopter with unknown time delay by utilizing an adaptive NN control approach, where NN was served as the uncertainty compensator. In \cite{6.Yang2020F}, combined with RBFNNs to estimate lumped disturbances, an adaptive backstepping control framework was constructed for the attitude control of a 3-DOF helicopter, where all parameters of the RBFNNs are trained online by means of the gradient descent algorithm. In \cite{7.li2021}, a command filter was introduced into the NN-based backstepping control framework to settle the complexity explosion problem in the design of the 3-DOF helicopter. Nevertheless, all the above-mentioned control strategies are asymptotic stability or uniformly ultimately boundedness, while the finite-time control strategy is more applicable for the attitude tracking control of the 3-DOF helicopter.

Recently, finite-time backstepping control has sparked extensive attention owing to the combination of the merits of conventional backstepping control and finite-time control scheme \cite{lemma5}. Due to the universal approximation capability of NNs and FLSs under certain conditions, the adaptive neural and fuzzy control methods \cite{2018neuro,2019neuro,2020neuro1,2020neuro2,2018fuzzy,2019fuzzy,2020fuzzy1,2020fuzzy2} are introduced into the finite-time backstepping control to cope with the uncertain nonlinearities existing in systems. However, the above-mentioned direct adaptive NN control approaches only consider the adaptation of weights, while the centers and widths of the RBFNN are set empirically and remain fixed during the entire control process, which may diminish the control quality \cite{2019GD}.

Motivated by the aforementioned discussion, a novel finite-time gradient descent-based adaptive NN finite-time control strategy is developed and applied to the attitude tracking control of a 3-DOF helicopter platform with composite disturbances in this work. First, a RBFNN is employed to approximate composite disturbances, where all the parameters of the RBFNN embracing the weights, centers and widths are trained online via finite-time gradient descent algorithm \cite{2006FGD,2020FGD}. Then, inspired by \cite{2020fuzzy2}, a finite-time backstepping control strategy is established to fulfill the tracking control of the elevation and pitch angles. In addition, a HFTD \cite{lemma6} is introduced for approximating the intermediate control signal and its derivative to prevent the problem of complexity explosion in the traditional backstepping design protocol, and the error generated by the HFTD is weakened by the corresponding compensation signal. The main features of the presented control strategy are concluded as follows:
\begin{enumerate}[1)]
	\item Compared with the adaptive NN control strategy where the centers and widths of the NN are artificially set, the proposed RBFNN estimator trains all parameters online by finite-time gradient descent algorithm, which further enhances the approximation capability of disturbances.
	\item By introducing the HFTD and error compensation system, the designed control strategy relaxes the assumption of the desired trajectory, while maintaining satisfactory finite-time tracking performance.
\end{enumerate}

With the help of the stability theorem, it is proved that the closed-loop system is semi-globally uniformly ultimately boundedness in finite time, while the attitude tracking errors can converge to an adequately small domain around the origin in finite time. A comparative simulation result is provided to delineate the effectiveness and advantages of the presented control strategy.

The remaining part of this paper is outlined as follows: Section II presents the system description and preliminaries. The finite-time gradient descent-based adaptive NN finite-time control strategy is designed in Section III. The comparative simulation experiment is implemented in Section V. Section VI gives the conclusion of this paper.

Notation: In this paper, $\left\|  \cdot  \right\|$ denotes the Euclidean norm and ${\mathop{\rm sig}\nolimits} {\left( x \right)^p} = {\left| x \right|^p}{\mathop{\rm sign}\nolimits} \left( x \right)$.
\section{System description and preliminaries}
\subsection{System description}
Fig 1 shows the structure of the 3-DOF helicopter platform. In view of elevation and pitch channels, the model of the platform can be formulated as the following form \cite{2020arXiv}

\begin{equation}
\begin{aligned}
{{\dot x}_1} &= {x_2}\\
{{\dot x}_2} &= \frac{{{L_a}}}{{{J_\alpha }}}{{\bar u}_1} - \frac{g}{{{J_\alpha }}}m{L_a}\cos ({x_1}) + {d_1}(x)\\
{{\dot x}_3} &= {x_4}\\
{{\dot x}_4} &= \frac{{{L_h}}}{{{J_\beta }}}{{\bar u}_2} + {d_2}(x)
\end{aligned}
\end{equation}
where ${x_1}$ and ${x_3}$ stand for elevation and pitch angles, respectively. ${d_1}\left( x \right)$ and ${d_2}\left( x \right)$ denote compound disturbances. The definitions and values of other variables can be seen in \cite{2020arXiv}. 

The control goal of this paper is to develop a finite-time gradient descent-based adaptive NN finite-time control strategy such that the tracking errors of the elevation angle and pitch angle can converge to an adequately small domain around the origin in finite time.

\begin{assumption}
The compound disturbances ${d_1}\left( x \right),{d_2}\left( x \right)$ are smooth and bounded.
\end{assumption}

\begin{assumption}
The reference signals ${x_{1r}}(t),{x_{3r}}(t)$ and their first derivative are known, smooth as well as bounded. 
\end{assumption}
\subsection{Preliminaries}
\begin{lemma} [\cite{lemma1}]
For ${y_j} \in {\bf{R}},j = 1,2, \cdots ,n,0 < \gamma \le 1$
\begin{equation}
{\left( {\sum\limits_{j = 1}^n {\left| {{y_j}} \right|} } \right)^\gamma } \le \sum\limits_{j = 1}^n {{{\left| {{y_j}} \right|}^\gamma } \le {n^{1 - \gamma }}} {\left( {\sum\limits_{j = 1}^n {\left| {{y_j}} \right|} } \right)^\gamma }
\end{equation}
\end{lemma}

\begin{lemma} [\cite{lemma2}]
For any ${\chi _1},{\chi _2} \in {\bf{R}}$, one has:
\begin{equation}
{\left| {{\chi _1}} \right|^{{c_1}}}{\left| {{\chi _2}} \right|^{{c_2}}} \le \frac{{{c_1}}}{{{c_1}{\rm{ + }}{{\rm{c}}_2}}}{\left| {{\chi _1}} \right|^{{c_1}{\rm{ + }}{{\rm{c}}_2}}}{\rm{ + }}\frac{{{c_2}}}{{{c_1}{\rm{ + }}{{\rm{c}}_2}}}{\left| {{\chi _2}} \right|^{{c_1}{\rm{ + }}{{\rm{c}}_2}}}
\end{equation}
where ${c_1} > 0,{c_2} > 0$ are constants.
\end{lemma}

To approximate the compound disturbances in the 3-DOF helicopter platform, the following lemma is required. 
\begin{lemma} [\cite{lemma3}]
For any continuous function $d\left( x \right)$ defined on a compact set $\Omega $, there exist $\bar \sigma  > 0$ and a RBFNN ${W^T}Q\left( x \right)$ satisfying
\begin{equation}
d\left( x \right) = {W^T}Q\left( x \right) + \sigma \left( x \right),{\rm{ }}\left| {\sigma \left( x \right)} \right| \le \bar \sigma 
\end{equation}
where $x = {\left[ {{x_1},{x_2}, \cdots ,{x_N}} \right]^T}$ and $W = {\left[ {{w_1},{w_2}, \cdots ,{w_M}} \right]^T}$ stand for the input vector and weight vector of the RBFNN, respectively. $Q\left( x \right) = {\left[ {{q_1}\left( x \right),{q_2}\left( x \right), \cdots ,{q_M}\left( x \right)} \right]^T}$ represents the basis function vector, where the basis function ${q_i}\left( x \right)$ is defined as
\begin{equation}
{q_i}\left( x \right) = \exp \left( { - \frac{{{{\left\| {x - {\mu _i}} \right\|}^2}}}{{\delta _i^2}}} \right)
\end{equation}
with ${\mu _i} = {\left[ {{\mu _{i1}},{\mu _{i2}}, \cdots ,{\mu _{iN}}} \right]^T}$ and ${\delta _i}$ being the center and width vectors of the basis function, respectively.
\end{lemma}

\section{Main Results}
In this section, the design process of the control law for the elevation angle tracking will be displayed in detail.
\subsection{Control Design}
Consider the elevation channel of the system (1):
\begin{equation}
\begin{aligned}
&{{\dot x}_1} = {x_2}\\
&{{\dot x}_2} = \frac{{{L_a}}}{{{J_\alpha }}}{{\bar u}_1} - \frac{g}{{{J_\alpha }}}m{L_a}\cos ({x_1}) + {d_1}(x)
\end{aligned}
\end{equation}

Then, the tracking errors can be defined as follows:
\begin{equation}
\begin{aligned}
&{e_1} = {x_1} - {x_{1r}}\\
&{e_2} = {x_2} - {x_{1,c}}
\end{aligned}
\end{equation}
where ${x_{1,c}}$ is the approximation of intermediate control input ${\alpha _m}$ by utilizing HFTD, which is formulated as [20].
\begin{equation}
\begin{aligned}
{{\dot x}_{1,c}} = & {x_{2,c}}\\
{\varepsilon ^2}{{\dot x}_{2,c}} = & - {a_0}\left( {{x_{1,c}} - {\alpha _m}} \right) - {a_1}{\mathop{\rm sig}\nolimits} {\left( {{x_{1,c}} - {\alpha _m}} \right)^{{r_1}}} \\
&- {b_0}\varepsilon {x_{2,c}} - {b_1}{\mathop{\rm sig}\nolimits} {\left( {\varepsilon {x_{2,c}}} \right)^{{r_2}}}
\end{aligned}
\end{equation}

According to lemma 6 in \cite{2021arXiv}, we have
\begin{equation}
{x_{1,c}} - {\alpha _m}\left( t \right) = {\rm O}\left( {{\varepsilon ^{{\rho _1}{r_2}}}} \right),{x_{2,c}} - {\dot \alpha _m}\left( t \right) = {\rm O}\left( {{\varepsilon ^{{\rho _1}{r_2} - 1}}} \right)\left( {t > {t_1}} \right)
\end{equation}

To compensate the error generated by HFTD, the error compensation signals ${\xi _1},{\xi _2}$ are constructed as follows [16]
\begin{equation}
\begin{aligned}
&{{\dot \xi }_1} =  - {k_1}{\xi _1} + {\xi _2} + \left( {{x_{1,c}} - {\alpha _m}} \right) - {n_1}\xi _1^h\\
&{{\dot \xi }_2} =  - {k_2}{\xi _2} - {\xi _1} - {n_2}\xi _2^h
\end{aligned}
\end{equation}
where ${k_1} > 0,{k_2} > 0,{n_1} > 0,{n_2} > 0$ are the appropriate design parameters. $h = {h_2}/{h_1} < 1$, ${h_1},{h_2}$ are positive odd numbers.

Define the compensated tracking errors ${z_1},{z_2}$ as the following form
\begin{equation}
\begin{aligned}
&{z_1} = {e_1} - {\xi _1}\\
&{z_2} = {e_2} - {\xi _2}
\end{aligned}
\end{equation}

Then, intermediate control input ${\alpha _m}$  and the controller ${\bar u_1}$ are developed as
\begin{equation}
\begin{aligned}
{\alpha _m} = & - {k_1}{e_1} + {{\dot x}_{1r}} - {m_1}z_1^h\\
{{\bar u}_1} = &\frac{{{J_\alpha }}}{{{L_a}}}\left( { - {k_2}{e_2} - {e_1} + {x_{2,c}} + \frac{g}{{{J_\alpha }}}m{L_a}\cos ({x_1})}\right.\\
 &\phantom{=\;\;}\left.{ - {m_2}z_2^h - {{W^T}Q\left( x \right)}} \right)
\end{aligned}
\end{equation}
where  ${m_1} > 0,{m_2} > 0$ are the appropriate design parameters and ${W^T}Q\left( x \right)$ denotes a RBFNN to estimate the composite disturbance ${d_1}(x)$. 

To train the parameters of the RBFNN, we construct the NN approximation error as follows
\begin{equation}
E = {\dot z_{2,c}} + {k_2}{z_2} + {z_1} + {m_2}z_2^h - {n_2}\xi _2^h
\end{equation}
where ${\dot z_{2,c}}$ is the approximation of ${\dot z_2}$ by adopting HFTD.

Design the cost function as $Y = 0.5{E^2}$, then in terms of finite-time gradient descent algorithm [18], [19], the iterative equations of the weights ${w_i}$, centers ${\mu _{i,j}}$ and widths ${\delta _i}$ can be given as
\begin{equation}
\begin{aligned}
&\Delta {w_i}\left( k \right) = E{q_i}\\
&{w_i}\left( {k + 1} \right) = {w_i}\left( k \right) + \lambda {\mathop{\rm sig}\nolimits} {\left( {\Delta {w_i}\left( k \right)} \right)^p}\\
&\Delta {\mu _{i,j}}\left( k \right) = E{w_i}{q_i}\frac{{\left( {{x_j} - {\mu _{i,j}}\left( k \right)} \right)}}{{\delta _i^2\left( k \right)}}\\
&{\mu _{i,j}}\left( {k + 1} \right) = {\mu _{i,j}}\left( k \right) + \lambda {\mathop{\rm sig}\nolimits} {\left( {\Delta {\mu _{i,j}}\left( k \right)} \right)^p}\\
&\Delta {\delta _i}\left( k \right) = E{w_i}{q_i}\frac{{{{\left\| {{x_j} - {\mu _i}\left( k \right)} \right\|}^2}}}{{\delta _i^3\left( k \right)}}\\
&{\delta _i}\left( {k + 1} \right) = {\delta _i}\left( k \right) + \lambda {\mathop{\rm sig}\nolimits} {\left( {\Delta {\delta _i}\left( k \right)} \right)^p}
\end{aligned}
\end{equation}
where $\lambda  \in \left( {0,1} \right)$ denotes the learning rate, $i = 1,2, \cdots ,M$, $j = 1,2, \cdots ,N$, $0 \le p < 1$.
\subsection{Stability Analysis}
\begin{theorem}
Consider the system (6) subject to Assumptions 1 and 2. If the HFTD is adopted as (8), the error compensation signals are constructed as (10), the intermediate control input is designed as (12), and the iterative equations of the parameters of the RBFNN are developed as (14), then the designed control laws ${\bar u_1}$ can guarantee that the closed-loop system is semi-globally uniformly ultimately boundedness in finite time, while the attitude tracking errors can converge to an adequately small domain around the origin in finite time.
\end{theorem}
\begin{proof}
By substituting (9), (10), (12) into (13) yields
\begin{equation}
\begin{aligned}
E &= {\dot z_{2,c}} + {k_2}{z_2} + {z_1} + {m_2}z_2^h - {n_2}\xi _2^h \\
&= {\dot z_2} - {\rm O}\left( {{\varepsilon ^{{\rho _2}{r_2} - 1}}} \right) + {k_2}{z_2} + {z_1} + {m_2}z_2^h - {n_2}\xi _2^h \\
&= {d_1} - {W^T}Q\left( x \right) - {\rm O}\left( {{\varepsilon ^{{\rho _2}{r_2} - 1}}} \right)
\end{aligned}
\end{equation}

In terms of \emph{lemma 3} and finite-time gradient descent algorithm, there exists ${t_2} > 0$ when $t > {t_2}$, one has
\begin{equation}
{d_1} - {\rm O}\left( {{\varepsilon ^{{\rho _2}{r_2} - 1}}} \right) = {W^T}Q\left( x \right) + \sigma ,{\rm{ }}\left| \sigma  \right| \le \bar \sigma 
\end{equation}

Consider the following Lyapunov function
\begin{equation}
V = \frac{1}{2}\sum\limits_{i = 1}^2 {\left( {z_i^2 + \xi _i^2} \right)} 
\end{equation}

Taking the time derivative of $V$, we have 
\begin{equation}
\begin{aligned}
\dot V = & - \sum\limits_{i = 1}^2 {\left( {{k_i}z_i^2 + {m_i}z_i^{1 + h} - {n_i}\xi _i^h{z_i}} \right)} \\
&- \sum\limits_{i = 1}^2 {\left( {{k_i}\xi _i^2 + {n_i}\xi _i^{1 + h}} \right) + \left( {{x_{1,c}} - {\alpha _m}} \right){\xi _1}} \\
&+ \left( {\sigma  + {\rm O}\left( {{\varepsilon ^{{\rho _2}{r_2} - 1}}} \right)} \right){z_2}
\end{aligned}
\end{equation}

Based on \emph{lemma 2}, when $t > \max \left\{ {{t_1},{t_2}} \right\}$, the following inequality holds
\begin{equation}
\begin{aligned}
\dot V \le  &- \left[ {{k_1}{z_1}^2 + \left( {{k_2} - \frac{1}{2}} \right){z_2}^2 + \left( {{k_1} - \frac{1}{2}} \right){\xi _1}^2 + {k_2}{\xi _2}^2} \right] \\
 &- \left[ {\sum\limits_{i = 1}^2 {\left( {{m_i} - \frac{{{n_i}}}{{1 + h}}} \right)z_i^{1 + h} + \sum\limits_{i = 1}^2 {\frac{{{n_i}}}{{1 + h}}\xi _i^{1 + h}} } } \right] + {\eta _3}
\end{aligned}
\end{equation}
where ${\eta _3} = 0.5{\rm O}\left( {{\varepsilon ^{2{\rho _1}{r_2}}}} \right) + 0.5{\left( {\bar \sigma  + {\rm O}\left( {{\varepsilon ^{{\rho _2}{r_2} - 1}}} \right)} \right)^2}$. 

According to \emph{lemma 1}, (19) can be further rewritten as
\begin{equation}
\dot V(x) \le  - {\eta _1}V(x) - {\eta _2}V{(x)^{\frac{{1 + h}}{2}}}{\rm{ + }}{\eta _3}
\end{equation}
where
\begin{equation}
\begin{aligned}
&{\eta _1} = \min \left\{ {\left( {2{k_1} - 1} \right),\left( {2{k_2} - 1} \right)} \right\}\\
&{\eta _2} = \min \left\{ {\left( {{m_i} - \frac{{{n_i}}}{{1 + h}}} \right){2^{\frac{{1 + h}}{2}}},\frac{{{n_i}}}{{1 + h}}{2^{\frac{{1 + h}}{2}}}} \right\}
\end{aligned}
\end{equation}

By adopting \emph{lemma 5} in [25], the closed-loop system is semi-globally uniformly ultimately boundedness in finite time. By designing the appropriate parameters, ${z_i},{\xi _i}$ can converge to adequately small domains around the origin in finite time $T$, which are given as follows
\begin{equation}
\begin{aligned}
&\left| {{z_i}} \right| \le \min \left\{ {\sqrt {\frac{{2{\eta _3}}}{{\left( {1 - \kappa } \right){\eta _1}}}} ,\sqrt {2{{\left( {\frac{{{\eta _3}}}{{\left( {1 - \kappa } \right){\eta _2}}}} \right)}^{\frac{2}{{1 + h}}}}} } \right\}\\
&\left| {{\xi _i}} \right| \le \min \left\{ {\sqrt {\frac{{2{\eta _3}}}{{\left( {1 - \kappa } \right){\eta _1}}}} ,\sqrt {2{{\left( {\frac{{{\eta _3}}}{{\left( {1 - \kappa } \right){\eta _2}}}} \right)}^{\frac{2}{{1 + h}}}}} } \right\}
\end{aligned}
\end{equation}
where $\kappa  \in \left( {0,1} \right)$

For $t \ge {T}$, the tracking error will stay in
\begin{equation}
\begin{aligned}
\left| {{e_1}} \right| &\le \left| {{z_1}} \right| + \left| {{\xi _1}} \right| \\
&\le \min \left\{ {2\sqrt {\frac{{2{\eta _3}}}{{\left( {1 - \kappa } \right){\eta _1}}}} ,2\sqrt {2{{\left( {\frac{{{\eta _3}}}{{\left( {1 - \kappa } \right){\eta _2}}}} \right)}^{\frac{2}{{1 + h}}}}} } \right\}
\end{aligned}
\end{equation}

The proof is completed.
\end{proof}

\section{Simulation Results}
This section conducts a contrastive simulation experiment on the tracking control of elevation angle to delineate the effectiveness and advantages of our developed control strategy. The initial condition of the elevation angle is ${x_1}\left( 0 \right) =  - {24^o}$ and the compound disturbance is selected as ${d_1} = \sin (2t)$. The target signal is given by
\begin{equation}
{x_{1r}}(t){\rm{ = }}\frac{\pi }{{18}}\sin (0.3\pi t - \frac{\pi }{2})
\end{equation}

Design the parameters of the constructed control strategy as ${k_1} = 1,{k_2} = 2,h = 0.6,{m_1} = {m_2} = 0.5,{n_1} = {n_2} = 1,{a_0} = 5,{a_1} = 0.5,\varepsilon  = 0.01,{b_0} = 2,{b_1} = 0.5,{r_1} = {r_2} = 0.5,p = 0.6,M = 5,\lambda  = 0.005$, and the adaptive NN backstepping control method in \cite{6.Yang2020F} is applied as the comparison approach. The comparison simulation results are shown in Figs. 2-3 where Fig. 2 depicts the curves of the attitude tracking error by adopting the developed control strategy as well as the NN backstepping control method, and Fig. 3 displays the curves of control input.
\begin{figure}\centering
	\includegraphics[width=0.4\textwidth]{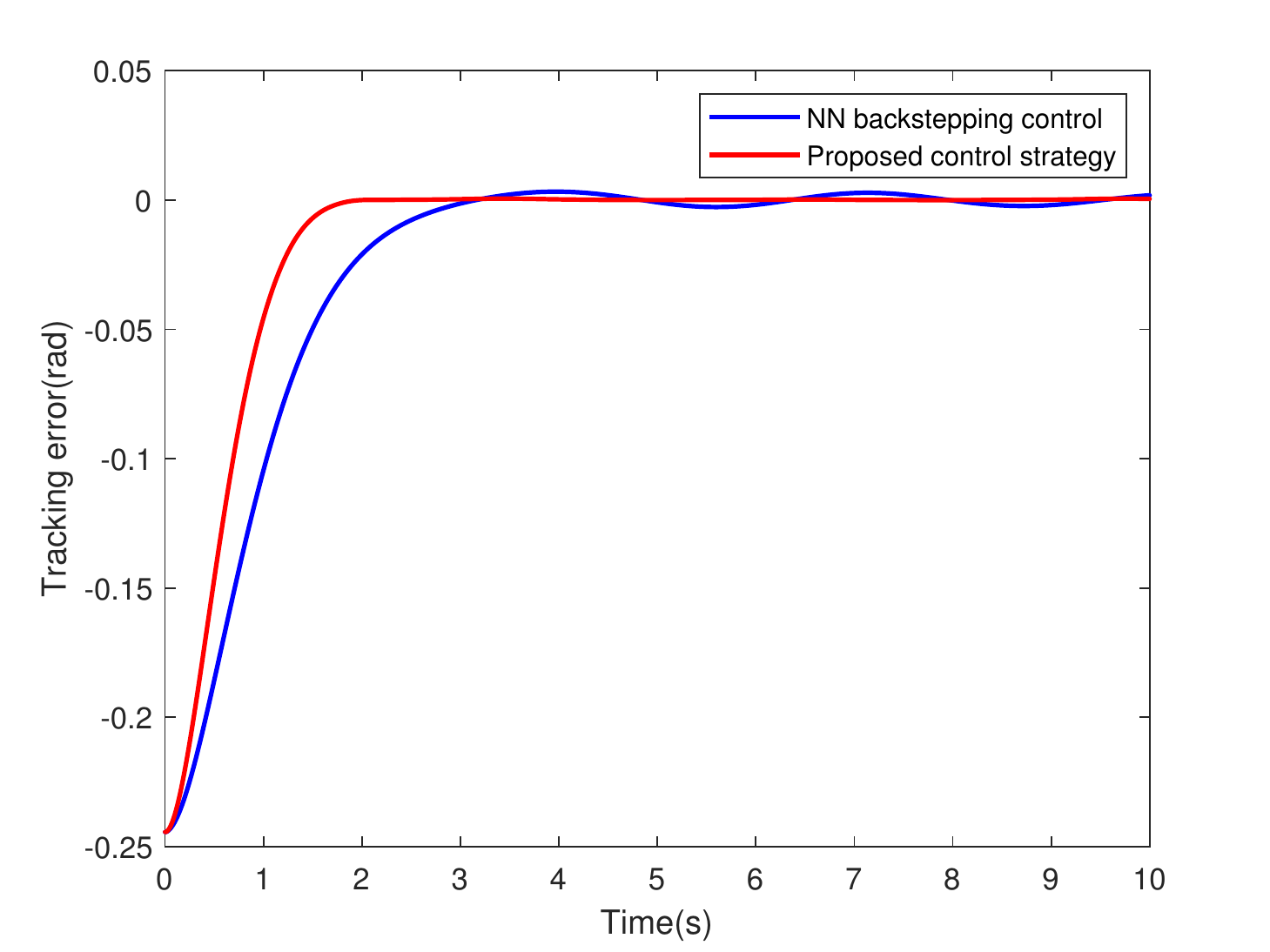}
	\caption{Tracking error}
\end{figure}

\begin{figure}\centering
	\includegraphics[width=0.4\textwidth]{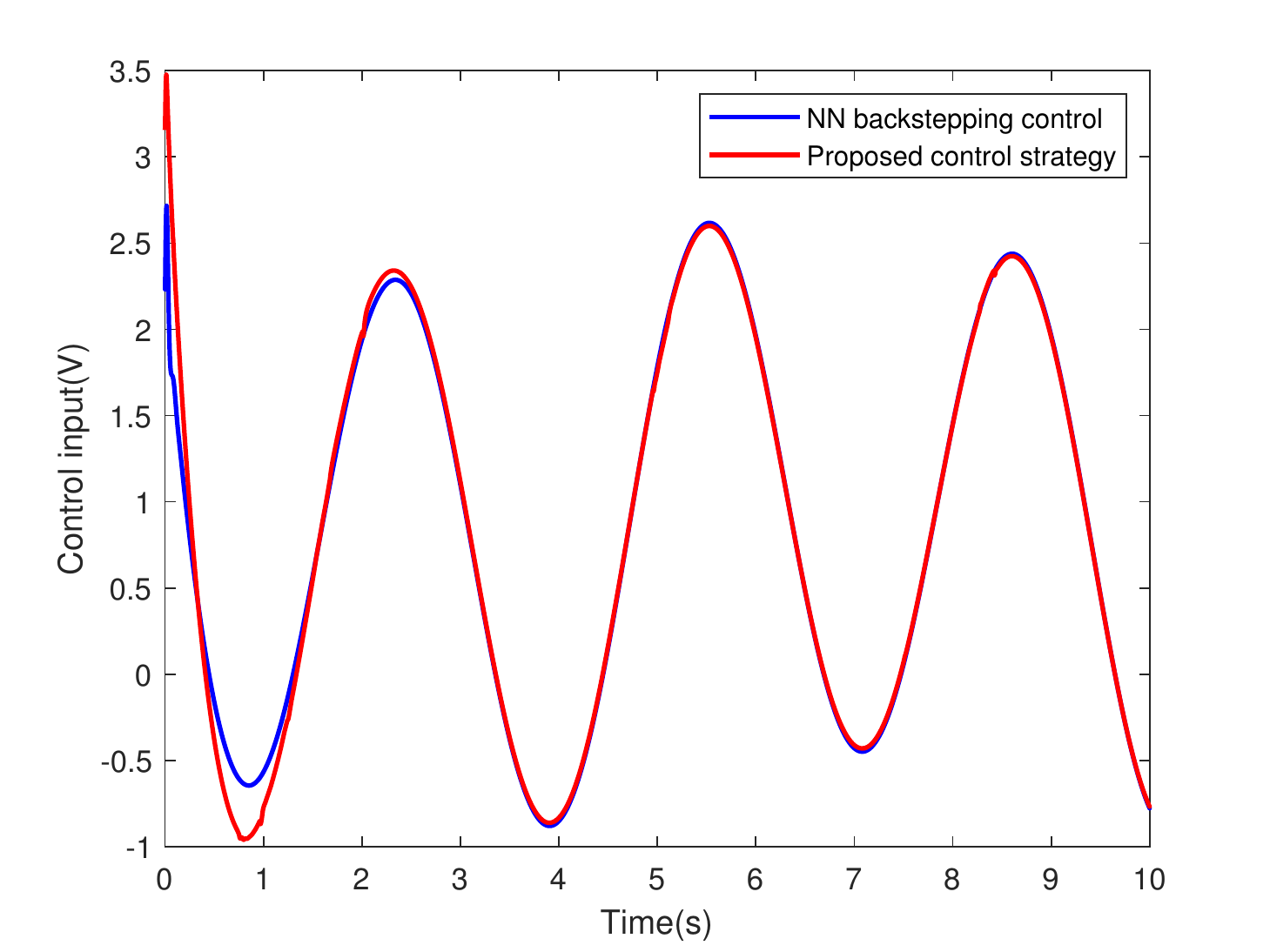}
	\caption{Control input ${\bar u}_1$}
\end{figure}

It is apparent from the simulation results that our developed control scheme can guarantee that the attitude tracking error converges to an adequately small domain around the origin in finite time. Furthermore, it can be observed that compared with the approach in [6], the control scheme developed by us possesses the features of faster convergence speed coupled with higher tracking accuracy.

\section{Conclusion}
In this work, a new finite-time gradient descent-based adaptive neural network finite-time control scheme has been presented to deal with the tracking problem of a 3-DOF helicopter platform with composite disturbances. A RBFNN is introduced to approximate composite disturbances, where all the parameters of the RBFNN are trained online by utilizing finite-time gradient descent algorithm. By integrating HFTD and error compensation system into the finite-time backstepping control framework, the problem of complexity explosion is addressed, and the influence of filter error is alleviated. The theoretical analysis illustrates that the closed-loop system is semi-globally uniformly ultimately boundedness in finite time. The effectiveness and advantages of the designed control strategy are well demonstrated by the comparative simulation.

\bibliographystyle{Bibliography/IEEEtranTIE}
\bibliography{Bibliography/IEEEabrv,Bibliography/myRef}\ 

\end{document}